\documentclass[runningheads]{llncs}

\usepackage{graphicx}
\usepackage{amsmath}
\usepackage{amssymb}
\usepackage{upgreek}
\usepackage{amsfonts}
\usepackage[colorlinks=true,linkcolor=red,citecolor=green,urlcolor=blue]{hyperref}
\usepackage{relsize}
\usepackage{algorithm,algpseudocode}

\floatstyle{ruled}
\newfloat{algorithm}{tbp}{loa}
\providecommand{\algorithmname}{Algorithm}
\floatname{algorithm}{\protect\algorithmname}

\algnewcommand{\Parameters}[1]{\item[\textbf{Parameters:}]{#1}}
\algnewcommand{\Init}[1]{\item[\textbf{Initialize:}]{#1}}

\DeclareMathOperator{\sgn}{sgn}

\begin{document}

\title{Uniswap Liquidity Provision: An Online Learning Approach}

\author{Yogev Bar-On\inst{1} \and Yishay Mansour\inst{1,2}}
\authorrunning{Bar-On \& Mansour}

\institute{Blavatnik School of Computer Science, Tel Aviv University, Tel Aviv, Israel \and
Google Research, Tel Aviv, Israel\\
\email{\{baronyogev,mansour.yishay\}@gmail.com}}

\maketitle

\begin{abstract}
Decentralized Exchanges (DEXs) are new types of marketplaces leveraging Blockchain technology. They allow users to trade assets with Automatic Market Makers (AMM), using funds provided by liquidity providers, removing the need for order books. One such DEX, Uniswap v3, allows liquidity providers to allocate funds more efficiently by specifying an active price interval for their funds. This introduces the problem of finding an optimal strategy for choosing price intervals. We formalize this problem as an online learning problem with non-stochastic rewards. We use regret-minimization methods to show a liquidity provision strategy that guarantees a lower bound on the reward. This is true even for non-stochastic changes to asset pricing, and we express this bound in terms of the trading volume.
\end{abstract}


\section{Introduction}
In recent years, there has been an increasing interest in Decentralized Finance (DeFi) since it opened the possibility for novel financial services that do not rely on a centralized entity, using Blockchain technology. One of the more prominent DeFi applications is \textit{Decentralized Exchanges} (DEXes), and specifically \textit{Automated Market Maker} (AMM) DEXes. Using AMMs, traders can trade between two or more assets without needing an entity on the other side of the trade to agree on a price.

Specifically, \textit{liquidity providers} (LPs) provide funds for a public pool of two assets. Traders can then send funds in the form of one asset to this pool, in exchange for funds in the form of a second asset. The pool automatically determines the price using only the data available on-chain, without outside sources, and includes a trading fee paid to the liquidity providers.

Uniswap, at the time of writing, is one of the largest AMM DEXes, holding reserves worth more than 3 Billion US dollars on Ethereum alone \cite{uniswapinfo}. Originally \cite{Adams2020UniswapVC}, liquidity providers on Uniswap provided two assets for a liquidity pool, and those funds could have been used by traders no matter how the price changed. This is not an efficient way of fund allocation, since a large portion of the funds would never be used, reserved for unrealistic prices.

Uniswap v3 \cite{adams2021uniswap} solves this problem by introducing a new type of AMM, which allows investors to allocate liquidity in specific price ranges, instead of allowing trades in any possible price. Using this concentrated liquidity, LPs can earn more money on their investment, given a correct prediction of market prices. However, when the market price is not in their invested range, LPs lose money on potential trades. Hence, there is great interest in strategies for choosing the correct price range for each investment.

Several attempts have been made to find the optimal liquidity provision strategy given some belief on the distribution of future asset prices. Still, no work has been done on finding an LP strategy for adversarial prices, where we do not assume anything about how the future price of an asset will change.

In this work, we consider the problem of liquidity provision in Uniswap v3 as an online learning problem in the non-stochastic setting. This allows us to formalize LP strategies as predictions with expert advice \cite{cesa1997use,cesa2006prediction}. We can construct a liquidity provision algorithm using regret-minimization methods that admit a positive reward for liquidity providers. This algorithm assumes almost nothing about how future asset prices will change. Moreover - we can express a lower bound on the reward using the trading volume during the investment time frame.

\subsection{Related Work}
Liquidity provision strategies on Uniswap v3 is an active research area. Neuder et al. \cite{neuder2021strategic} studied optimal strategies where prices evolve according to a Markov chain, and Fan et al. \cite{fan2022differential} expanded on this work with a better analysis that also takes impermanent loss into account. Heimbach et al. \cite{heimbach2022risks} also formalize the problem and strategy analysis in a Black-Scholes stochastic market model. Fritch \cite{fritsch2021concentrated} measures the performance of liquidity providers based on empirical data. More recently, Milionis et al. \cite{milionis2022quantifying} investigate the rewards of liquidity providers in an idealized Black-Scholes stochastic market, and provide a predictable running cost of providing liquidity. Cartea et al. \cite{cartea2022decentralised} show an optimal liquidity provision strategy where the reward is compared to a self-financing portfolio in a risk-free account.

Regarding liquidity pool mechanics, Angeris et al. \cite{angeris2019analysis} showed that Uniswap v2 pools closely track real asset prices under rational assumptions. This result is generalized for more settings in \cite{angeris2022constant,angeris2020improved}. Engel and Herlihy \cite{engel2021composing} formalized and studied the effects of composing multiple AMMs for a trade.

Using online learning for market making was already considered for traditional finance. Chen and Vaughan \cite{chen2010new} showed the connection between cost-function based prediction markets and no-regret learning. Penna and Reid \cite{della2011bandit} studied using bandit algorithms for automatic market making. Abernethy et al. \cite{abernethy2013efficient} provided a general optimization framework for the design of securities markets, and \cite{abernethy2013adaptive} showed an adaptive market-making algorithm based on the order-book spread, using online learning techniques. More recently, Spooner et al. \cite{spooner2018market} improved on those results using a more generic reinforcement learning framework.

\subsection{Outline}
In the next section, we overview AMM mechanisms on Uniswap. We then formally define our model for online liquidity provision strategies in Section \ref{sec:model}. Section \ref{sec:static_strategies} analyzes the reward of a special class of those strategies called static strategies. We use those strategies to present and analyze a dynamic regret-minimization strategy in Section \ref{sec:adaptive_strategy}. We then simulate our model with empirical data in Section \ref{sec:simulations}. We conclude our work in Section \ref{sec:conclusions}.


\section{Uniswap Overview}\label{sec:uniswap}
Any Uniswap liquidity pool holds two asset reserves, say token A reserve and token B reserve. Those assets are provided by liquidity providers, and traders can use the liquidity pool to trade between the two tokens.

\subsection{Constant Product Market Makers}
When trades are made on Uniswap v2, the pool maintains an invariant on the number of tokens it holds. Let $x>0$ and $y>0$ be the amounts of token A and token B in the pool, respectively. The amounts must satisfy $F(x,y)\triangleq xy=L^2$, where $L>0$, the pool's liquidity, cannot be changed by trades. This mechanism is called a Constant Product Market Maker (CPMM). Say the trader swaps $\Updelta x$ units of token A. They will receive $\Updelta y$ units of token B, such that $F(x+\Updelta x, y - \Updelta y)=L^2$. The motivation is supply \& demand: when a token is bought, its price goes up, and when it is sold, its price goes down.

\subsubsection{Pricing}
The \textit{spot price} is defined as the ratio of the amount of token B received to the amount of token A sent, for an infinitesimal trade: 
\[ p=\lim_{\Updelta x\to 0} \frac{\Updelta y}{\Updelta x}.\]
We can express this using the current reserves:
\[y-\Updelta y = \frac{L^2}{x+\Updelta x} = \frac{xy}{x+\Updelta x} \]
\[\frac{\Updelta y}{\Updelta x} = y \left(\frac{1}{\Updelta x} - \frac{x}{\Updelta x (x+\Updelta x)}\right) = \frac{y}{x+\Updelta x} \]
\[p=\lim_{\Updelta x\to 0} \frac{\Updelta y}{\Updelta x} = \frac{y}{x}.\]
and thus the spot price is simply the ratio of the pool's reserves. This means that as $L$ grows larger (more funds in the reserves), it takes a larger trading volume to affect the price. Also, you can see that for very large trades, the effective price can be much different from the spot price - this is called \textit{slippage}.

\subsubsection{Providing Liquidity}
Liquidity providers can always add liquidity to the pool, by providing a bundle of both token A and token B, such that the spot price is not changed. Since $xy=L^2$ and $p=\frac{y}{x}$, we get that:
\begin{align}\label{eq:uni2_reserves}
x=\frac{L}{\sqrt{p}}, && y=L\sqrt{p}.
\end{align}
Hence the liquidity grows linearly by the number of new tokens provided by an LP. If a liquidity provider added $x',y'$ tokens to the pool, such that $\frac{y'}{x'}=p$ and $x' y' = L'^2$, the total liquidity in the pool grows by $L'$. LPs can always withdraw their funds from the pool's reserve, but the amounts are proportional to the liquidity they provided to the pool and not the same as they initially provided. They can withdraw a bundle of $\frac{L'}{L+L'}x$ units of token A and $\frac{L'}{L+L'}y$ units of token B. This means price changes affect liquidity providers, and they might lose money from price volatility (compared to a situation where they did not invest at all); this is called \textit{impermanent loss}.

\subsubsection{Trading Fees}
The incentive for LPs is that when a trade is made, e.g., from token A to token B, the trader only receives $(1-\gamma)\Updelta y$ units of token B, where $\gamma \Updelta y$ is the \textit{trading fee} that goes towards liquidity providers. $\gamma$ is the \textit{fee tier} of the pool, and is usually between $0.05\%$ to $1\%$. The trading fee is distributed between all liquidity providers, proportional to their liquidity in the pool. Hence, if an LP provided $L'$ liquidity units to a pool with $L$ total units of liquidity, they will receive $\gamma\Updelta y \frac{L'}{L}$ units of token B as a trading fee. 

\subsection{Concentrated Liquidity}
A Uniswap v2 pool reserve will always have available funds, no matter the spot price. If a liquidity provider believes the spot price of a liquidity pool is contained in a limited interval, the liquidity allocation is wasteful. To be more efficient, Uniswap v3 introduces \textit{concentrated liquidity}.

In Uniswap v3, the spot price domain is divided into discrete intervals: $[d^i, d^{i+1})$ for all $i \in \mathbb{Z}$, where $d>1$ is the interval size (usually $1.0001\le d \le 1.01$, varies by a \textit{tick spacing} parameter). Each spot price interval acts as an independent CPMM liquidity pool, and liquidity providers can choose to provide liquidity only for certain intervals. This means LP's funds are allocated more efficiently, and they get more trading fees per tokens provided. However, if the spot price is outside the active interval of an LP, they will not earn any trading fees.

To achieve this, Uniswap v3 keeps track of \textit{virtual reserves} (see Figure \ref{fig:uni3_virtual}). Say LPs invest $x$ units of token A and $y$ units of token B in the price range $[a,b]$. The virtual reserves $x_v,y_v$ are defined to hold:
\begin{align*}
    x_v = x + \frac{L}{\sqrt{b}}, && y_v = y + \sqrt{a}L,
\end{align*}
where $L=\sqrt{x_v y_v}$ is the amount of liquidity provided. Basically, those are the reserves needed to get the same liquidity in a Uniswap v2 liquidity pool. Notice that when providing liquidity for $(0, \infty)$, the virtual reserves coincide with the real reserves, and thus this is equivalent to investing in a Uniswap v2 pool.
\begin{figure}[ht]
\includegraphics[width=0.67\textwidth]{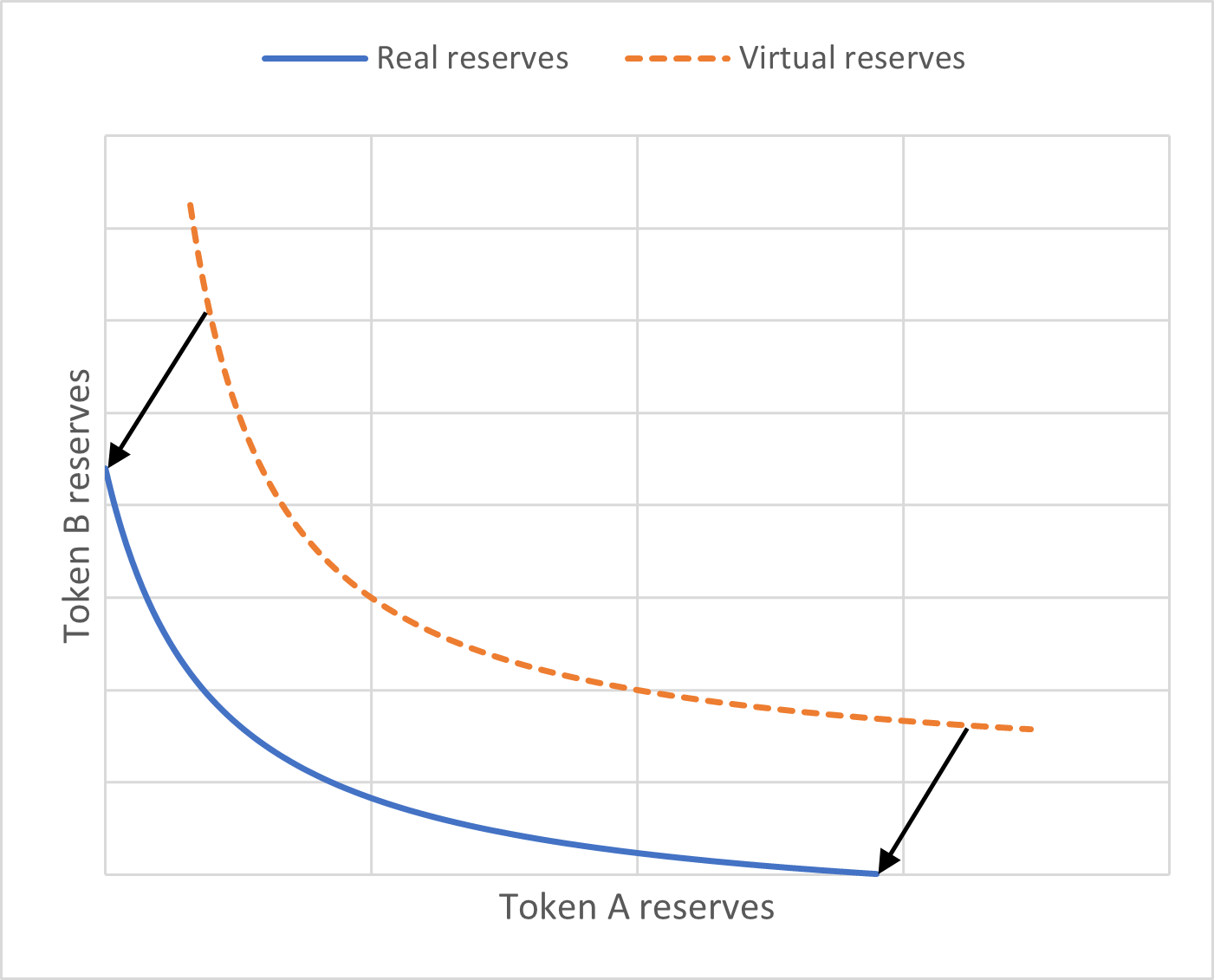}
\centering
\caption{Uniswap v3 virtual reserves illustration.}\label{fig:uni3_virtual}
\end{figure}

The trade price can then be computed using the virtual reserves as if it were a Uniswap v2 pool. From Eq. \ref{eq:uni2_reserves} we get that $x_v = \frac{L}{\sqrt{p}}$ and $y_v = \sqrt{p}L$, and thus:
\begin{align}\label{eq:uni3_reserves}
x = (\frac{1}{\sqrt{p}} - \frac{1}{\sqrt{b}})L, && y = (\sqrt{p}-\sqrt{a})L
\end{align}

What happens when the spot price leaves the active range? In the case $p>b$, we would have no real reserve of token A, and $(\sqrt{b}-\sqrt{a})L$ units of token B. In the same way if $p<a$, we would get $(\frac{1}{\sqrt{a}}-\frac{1}{\sqrt{b}})L$ units of token A. This can be a way for LPs to make limit-order trades, trading their entire investment to a single token once the price leaves a certain range.


\section{Online Learning Model}\label{sec:model}
We consider a Uniswap v3 Liquidity Pool, consisting of token A and token B with an interval size $d\in [1.0001,2)$ and a fee tier $\gamma < 1$. We now present the online learning strategy for liquidity provision, where at each time step, we re-invest our entire portfolio in a price range centered at the current spot price.

For each time step $t=1,...,T$, we invest a portfolio worth $M_t$ units of token B, consisting of $Q_t^{(A)} = \frac{M_t}{2 p_t}$ units of token A and $Q_t^{(B)} = \frac{M_t}{2}$ units of token B, where $p_t$ is the spot price of token A (in terms of token B) at step $t$; we pick a positive \textit{concentration controller} $n_t\in\mathbb{N}\cup\{\infty\}$ and provide liquidity in the range $[p_t d^{-n_t}, p_t d^{n_t}]$ (notice this assumes the interval size is fine enough to approximate spot prices as its exponents). Using Eq. \ref{eq:uni3_reserves}, at each step $t$ we thus have an active liquidity of
\[ L_t = \frac{Q_t^{(B)}}{\sqrt{p_t} - \sqrt{p_t d^{-n_t}}} = \frac{M_t}{2 \sqrt{p_t} (1-d^{-\frac{n_t}{2}})}. \]

When the step is over, a new price $p_{t+1}=p_t d^{\rho_t}$ (where $\rho_t \triangleq \log_d{\frac{p_{t+1}}{p_t}}$ is the \textit{logarithmic price change}) is revealed, along with the trading volume $v_t$. We then receive some reward $r_t(n_t)=\ln{\frac{M_{t+1}}{M_t}}$, re-balance our funds to match the new price, and invest again. Our goal is to choose optimal concentration controllers for maximizing the total reward:
\[ G_T(n_1,...,n_T) = \ln{\frac{M_T}{M_1}} = \sum_{t=1}^T{r_t(n_t)}. \]

Note that the reward is defined in terms of token B, so it holds a hidden assumption that we care (w.l.o.g.) about maximizing our holdings in terms of token B (e.g., maximizing the number of dollars we have rather than the number of euros).

\subsection{Assumptions}

We first make a few assumptions to simplify our model and analysis.

\subsubsection{Total Liquidity}
We assume the total active liquidity $\mathbb{L}_t$ at time step $t$ remains constant throughout the step, and we define $K_t$ as the equivalent virtual reserves for that amount of liquidity. Hence we get from Eq. \ref{eq:uni2_reserves}:
\[K_t = 2\sqrt{p_t}\,\mathbb{L}_t. \]

It is useful to express the volume in terms of those virtual reserves, leading us to define the \textit{relative trading volume} as the relation between the trading volume and the active virtual reserves:
\[ u_t \triangleq \frac{v_t}{K_t} = \frac{v_t}{2\sqrt{p_t}\,\mathbb{L}_t}. \]

\subsubsection{Partial Fees}
We assume either the price is in our active range the entire step or not at all - meaning we get all the commission or none of it. In reality, the price can come in and out of the active range mid-step since we assume many trades can be made on the same step.

\subsubsection{Gas Price}
Gas prices are equivalent to adding a negative term, linear in $T$, to the total reward. To simplify our results, we ignore this term and assume re-balancing and re-investing our funds does not incur a loss.

\subsection{Reward Function}
We will now show the form of the reward function at step $t$, given a concentration controller $n_t$, trading volume $v_t$, and price change $\rho_t$. Trivially, the reward consists of the trading fees $f_{v_t}(n_t)$ we get as a commission. But it also depends on our portfolio's value at the end of the step (before re-balancing), $M_{t\rightarrow t+1}$, which is different than the beginning. Hence:
\[ M_{t+1} = M_{t\rightarrow t+1} + f_{v_t}(n_t),\]
and we define the reward as the logarithmic change of value, measured in units of token B:
\[ r_t(n_t) \triangleq \ln{\frac{M_{t+1}}{M_t}} = \ln{\left(M_{t\rightarrow t+1} + f_{v_t}(n_t)\right)} - \ln{M_t}.\]

\subsubsection{Trading Fees}
The trading fee at step $t$ is linearly proportional to the trading volume $v_t$ and our active liquidity, and inversely proportional to the total active liquidity in the pool, with the trading fee $\gamma$ as the proportionality constant. Hence our commission is:
\[ f_{v_t}(n_t) = \begin{cases}
0 & |\rho_t| > n_t \\
\gamma v_t \frac{L_t}{\mathbb{L}_t} & \textrm{else}
\end{cases} = \begin{cases}
0 & |\rho_t| > n_t \\
\gamma u_t \frac{M_t}{( 1- d^{-\frac{n_t}{2}} ) } & \textrm{else}
\end{cases}.
\]

\subsubsection{Change in Value}

When we invest at step $t$ like our strategy suggests, our portfolio will consist of the following amount of tokens (see Eq. \ref{eq:uni3_reserves}):
\[
Q_{t\rightarrow t+1}^{(A)} =
\begin{cases}
0 & \rho_t > n_t \\
L_t \frac{1}{\sqrt{p_{t}}} (d^{\frac{n_t}{2}}-d^{-\frac{n_t}{2}}) & \rho_t < -n_t  \\
L_t \frac{1}{\sqrt{p_{t}}} (d^{-\frac{\rho_t}{2}}-d^{-\frac{n_t}{2}}) & \textrm{else}
\end{cases}
\]
\[
Q_{t\rightarrow t+1}^{(B)} =
\begin{cases}
L_t \sqrt{p_{t}} (d^{\frac{n_t}{2}}-d^{-\frac{n_t}{2}}) & \rho_t > n_t \\
0 & \rho_t < -n_t  \\
L_t \sqrt{p_{t}} (d^{\frac{\rho_t}{2}}-d^{-\frac{n_t}{2}}) & \textrm{else}
\end{cases}
\]
hence, at the end of step $t$, our portfolio will be worth:
\begin{flalign*}
M_{t\rightarrow t+1} &= p_{t+1}Q_{{t\rightarrow t+1}}^{(A)} + Q_{{t\rightarrow t+1}}^{(B)} \\
&= \begin{cases}
L_t \sqrt{p_{t}} (d^{\frac{n_t}{2}}-d^{-\frac{n_t}{2}}) & \rho_t > n_t \\
L_t \sqrt{p_{t}} d^{\rho_t} (d^{\frac{n_t}{2}}-d^{-\frac{n_t}{2}}) & \rho_t < -n_t \\
L_t \sqrt{p_t} \left( 2d^\frac{\rho_t}{2} - d^{-\frac{n_t}{2}} (1 + d^{\rho_t} )\right) & \textrm{else}
\end{cases} \\
&= \begin{cases}
M_t \frac{1 + d^{\frac{n_t}{2}}}{2} & \rho_t > n_t \\
M_t \frac{d^{\rho_t} (1 + d^{\frac{n_t}{2}})}{2} & \rho_t < -n_t \\
M_t \frac{2d^\frac{\rho_t}{2} - d^{-\frac{n_t}{2}} (1 + d^{\rho_t})}{2 ( 1- d^{-\frac{n_t}{2}} )} & \textrm{else}
\end{cases}.
\end{flalign*}

\subsubsection{Total Reward}

Summarizing, we get that:
\[ \frac{M_{t+1}}{M_t} =  \frac{M_{t\rightarrow t+1} + f_{v_t}(n_t)}{M_t} = \begin{cases}
\frac{1 + d^{\frac{n_t}{2}}}{2} & \rho_t > n_t \\
\frac{d^{\rho_t} (1 + d^{\frac{n_t}{2}})}{2} & \rho_t < -n_t \\
\frac{2d^\frac{\rho_t}{2} - d^{-\frac{n_t}{2}} (1 + d^{\rho_t}) + 2\gamma u_t}{2 ( 1- d^{-\frac{n_t}{2}} )} & \textrm{else}
\end{cases},
\]
and thus:
\begin{equation}\label{eq:reward}
r_t(n_t) = \ln{\frac{M_{t+1}}{M_t}} = \begin{cases}
\ln{\frac{1 + d^{\frac{n_t}{2}}}{2}} & \rho_t > n_t \\
\ln{\frac{d^{\rho_t} (1 + d^{\frac{n_t}{2}})}{2}} & \rho_t < -n_t \\
\ln{\frac{2d^\frac{\rho_t}{2} - d^{-\frac{n_t}{2}} (1 + d^{\rho_t}) + 2\gamma u_t}{2 ( 1- d^{-\frac{n_t}{2}} )}} & \textrm{else}
\end{cases}.
\end{equation}

As expected, using smaller liquidity intervals leads to larger rewards, but increases the chances of losing trading fees. More specifically:
\begin{itemize}
\item If $p_t>n_t$, there are no trading fees, but the reward is larger than $1$ since the value of token A is larger (in terms of token B). Once the price passed the active interval, all tokens were converted to type B - so the reward does not depend on the price, only on the concentration controller.
\item If $p_t<-n_t$, there are no trading fees, and the portfolio's value is smaller, so the reward is less than $1$. Here, the reward does depend on the price, since all tokens are converted to type A (while we measure the value in terms of token B).
\item If $-n_t\le p_t \le n_t$, the reward might be either greater or less than $1$. The term $2d^\frac{\rho_t}{2} - d^{-\frac{n_t}{2}} (1 + d^{\rho_t})$ is the effect of the price change on the reward, and $\gamma u_t$ is the effect of fees collected. The term $1- d^{-\frac{n_t}{2}}$ in the denominator signifies the fact that as the liquidity is concentrated in a smaller interval, the reward is larger.
\end{itemize}


\section{Static Strategies}\label{sec:static_strategies}

In this section, we will analyze $n$-static strategies in which we pick a constant concentration controller $n$ for all steps. We will denote the total reward for static strategies as $G_T(n) = G_T(n,...,n)$, and we will bound it in terms of the average logarithmic price change:
\[ P \triangleq \frac{1}{T}\sum_{t=1}^T{|\log_d{p_{t+1}} - \log_d{p_{t}}|} = \frac{1}{T}\sum_{t=1}^T{|\rho_t|}, \]
which is essentially the is the average number of price intervals skipped each step.

\subsection{$n=\infty$}
We first consider the $\infty$-static strategy, which will be useful as a baseline. This is equivalent to investing in a Uniswap v2 pool while re-balancing funds in each step.

\begin{lemma}\label{lem:v2_bound}
For any sequence of $\rho_t,u_t$ s.t. $u_t \leq \frac{2}{\gamma}$, the reward of the $\infty$-static strategy holds:
\[ G_T(\infty) \geq \frac{\gamma}{2}\sum_{t=1}^T{u_t} - \frac{T}{2}P \ln{d}. \]
\end{lemma}

\begin{proof}
Taking the limit where $n_t \rightarrow \infty$ in Eq. \ref{eq:reward} and getting $d^{-\frac{n_t}{2}}=0$, we have:
\begin{flalign*}
G_T(\infty) &= \sum_{t=1}^T{r_t(\infty)} = \sum_{t=1}^T{\ln\left(d^\frac{\rho_t}{2} + \gamma u_t\right)}
\geq \sum_{t=1}^T{\ln\left(d^{-\frac{|\rho_t|}{2}} + \gamma u_t\right)} \\
&\geq \sum_{t=1}^T{\ln\left(d^{-\frac{|\rho_t|}{2}}\left(1 + \gamma u_t\right)\right)}
= \sum_{t=1}^T{\ln\left( 1 + \gamma u_t \right)} - \frac{\ln{d}}{2}\sum_{t=1}^T{|\rho_t|} \\
&= \sum_{t=1}^T{\ln\left( 1 + \gamma u_t \right)} - \frac{T}{2}P \ln{d}
\geq \frac{\gamma}{2}\sum_{t=1}^T{u_t} - \frac{T}{2}P \ln{d}
\end{flalign*}
as desired, where the last inequality is since $\ln(1+x) \geq \frac{x}{2}$ for $x \leq 2$.
\qed\end{proof}

We can already see a lower bound on the trading volume which will ensure we get a positive reward:

\begin{corollary}\label{cor:v2_positive}
For any sequence of $\rho_t,u_t$ s.t. $u_t \leq \frac{2}{\gamma}$ and $\frac{1}{T}\sum_{t=1}^T{u_t} \geq \frac{1}{\gamma}P \ln{d}$, the $\infty$-static strategy reward is positive:
\[ G_T(\infty) \geq 0. \]
\end{corollary}

\subsection{$n < \infty$}
For analysis of the case where $n$ is finite, we denote by $T_{\leq n}$ the set of time steps where the price remained in the active interval, and by $T_{>n}$ the set of steps where the price left the active interval, i.e.,
\begin{align*}
T_{\leq n}=\lbrace t \mid |\rho_t|\leq n\rbrace  &&  T_{>n}=\lbrace t \mid |\rho_t|>n\rbrace.
\end{align*}

We show the following lower bound on the reward,
\begin{lemma}\label{lem:static_bound}
The reward of the $n$-static strategy holds:
\[ G_T(n) \geq \frac{n}{4}|T_{>n}|\ln{d} - T P\ln{d} + \sum_{|\rho_t|\leq n}{\ln\left( 1 + \frac{\gamma u_t}{1-d^{-\frac{n}{2}}}\right)}. \]
\end{lemma}

\begin{proof}
First, we will show a new representation of the total reward:
\begin{flalign}
G_T(n) &= \sum_{t=1}^T{r_t(n)} \\\label{eq:static_bound_1}
&= \sum_{\rho_t>n}{\ln\left( \frac{1+d^\frac{n}{2}}{2} \right)} + \sum_{\rho_t<-n}{\ln\left( d^{\rho_t} \frac{1+d^\frac{n}{2}}{2} \right)} + \sum_{|\rho_t|\leq n}{\ln\left( \frac{2d^\frac{\rho_t}{2} + 2\gamma u_t - d^{-\frac{n}{2}}(1+d^{\rho_t})}{2(1-d^{-\frac{n}{2}})} \right)} \\\label{eq:static_bound_2}
&= |T_{>n}| \ln\left( \frac{1+d^\frac{n}{2}}{2} \right) - \ln{d}\sum_{\rho_t<-n}{|\rho_t|} + \sum_{|\rho_t|\leq n}{\ln\left( \frac{2d^\frac{\rho_t}{2} + 2\gamma u_t - d^{-\frac{n}{2}}(1+d^{\rho_t})}{2(1-d^{-\frac{n}{2}})} \right)} \\\label{eq:static_bound_3}
&= |T_{>n}| \ln\left( \frac{1+d^\frac{n}{2}}{2} \right) - \ln{d}\sum_{\rho_t<n}{|\rho_t|} + \sum_{|\rho_t|\leq n}{\ln\left(d^{|\rho_t|} \frac{2d^\frac{\rho_t}{2} + 2\gamma u_t - d^{-\frac{n}{2}}(1+d^{\rho_t})}{2(1-d^{-\frac{n}{2}})} \right)},
\end{flalign}
where (\ref{eq:static_bound_1}) follows directly from (\ref{eq:reward}); in (\ref{eq:static_bound_2}) we use the fact that for negative $\rho_t$:
\[\ln\left( d^{\rho_t} \frac{1+d^\frac{n}{2}}{2} \right) = \ln\left(\frac{1+d^\frac{n}{2}}{2} \right)-|\rho_t|\ln{d},\] and in (\ref{eq:static_bound_3}) we simply add to middle summation the term $\ln{d}\sum_{|\rho_t|<n}{|\rho_t|}$, compensating it in the last summation by multiplying by $d^{|\rho_t|}$ inside the logarithm.

Since we look for a lower bound in terms of the average logarithmic price change $P$, we can complete the middle summation over all time steps, getting $TP\ln{d}$, and separate the volume term in the last summation using the fact that $d^{|\rho_t|}\geq 1$:
\begin{flalign}
G_T(n) &\geq |T_{>n}| \ln\left( \frac{1+d^\frac{n}{2}}{2} \right) - TP\ln{d} + \sum_{|\rho_t|\leq n}{\ln\left( d^{|\rho_t|} \frac{2d^\frac{\rho_t}{2} - d^{-\frac{n}{2}}(1+d^{\rho_t})}{2(1-d^{-\frac{n}{2}})} + \frac{\gamma u_t}{1-d^{-\frac{n}{2}}} \right)} \\\label{eq:static_bound_4}
&\geq |T_{>n}| \frac{n\ln{d}}{4} - TP\ln{d} + \sum_{|\rho_t|\leq n}{\ln\left( d^{|\rho_t|} \frac{2d^\frac{\rho_t}{2} - d^{-\frac{n}{2}}(1+d^{\rho_t})}{2(1-d^{-\frac{n}{2}})} + \frac{\gamma u_t}{1-d^{-\frac{n}{2}}} \right)}
\end{flalign}
where (\ref{eq:static_bound_4}) follows from $\ln\frac{1+e^{x}}{2} \geq \frac{x}{2}$.

Now, consider the term $f(\rho_t)=d^{|\rho_t|}\left(2d^\frac{\rho_t}{2} - d^{-\frac{n}{2}}(1+d^{\rho_t})\right)$ as a function of $\rho_t$. We want to find a lower bound on the domain $|\rho_t|<n$; we first look for a local minimum on the differentiable domain:
\[
f'(\rho_t) = d^{|\rho_t|}\ln{d}\left(\sgn(\rho_t)\left(2d^\frac{\rho_t}{2} - d^{-\frac{n}{2}}(1+d^{\rho_t})\right)
+ d^\frac{\rho_t}{2} - d^{\rho_t-\frac{n}{2}}\right) = 0 \\
\]
Hence,
\[ (2\sgn(\rho_t)+1)d^{\frac{n+\rho_t}{2}} - (1+\sgn(\rho_t))d^{\rho_t} - \sgn(\rho_t) = 0. \]
In the case that $\rho_t < 0$, we get $d^{\frac{n+\rho_t}{2}} - 1 = 0$, which is not solvable when $|\rho_t| < n$. For $\rho_t > 0$, we get $2d^{\rho_t} - 3d^{\frac{n+\rho_t}{2}} + 1 = 0$, which is a quadratic equation on $d^{\frac{\rho_t}{2}}$ with the solution:
\[  d^{\frac{\rho_t}{2}} = \frac{1}{4} (3d^{\frac{n}{2}} \pm \sqrt{9d^{n} - 8}). \]
Note that $\sqrt{9d^{n} - 8} \geq d^{\frac{n}{2}}$, so we get $d^{\frac{\rho_t}{2}} \geq d^{\frac{n}{2}}$ for the first solution, which is outside the domain. Also,
\[ \sqrt{9d^{n} - 8} = \sqrt{(3d^{\frac{n}{2}} - 4)^2 + 24(d^{\frac{n}{2}}-1)} \geq 3d^{\frac{n}{2}} - 4 \]
so we get $d^{\frac{\rho_t}{2}} \leq 1$ for the second solution, which is impossible since we assumed $\rho_t > 0$.

In conclusion, there is no local minimum on the differentiable domain, and $f(\rho_t)$ admits a lower bound either on the domain's boundary ($\rho_t = \pm n$) or at $\rho_t = 0$ (the only non-differentiable point). Evaluating $f$ at each of those points we get:
\begin{align*}
f(0) = 2(1-d^{-\frac{n}{2}}), && f(n) = d^{\frac{3n}{2}} - d^{\frac{n}{2}}, && f(-n) = d^{\frac{n}{2}} - d^{-\frac{n}{2}}.
\end{align*}
Note that:
\[f(-n) - f(0) = 2\left(\frac{d^\frac{n}{2} + d^{-\frac{n}{2}}}{2} - 1\right) \geq 0,\]
so $f(-n)\geq f(0)$. We can also see that $f(n)=f(-n)d^n\geq f(-n)$. To summarize, we have $f(n)\geq f(-n) \geq f(0)$, and thus:
\begin{equation}\label{eq:term_bound}
d^{\frac{3n}{2}} - d^{\frac{n}{2}} \geq d^{|\rho_t|}\left(2d^\frac{\rho_t}{2} - d^{-\frac{n}{2}}(1+d^{\rho_t})\right) \geq 2(1-d^{-\frac{n}{2}}).
\end{equation}
We can now plug this result in (\ref{eq:static_bound_4}) to conclude our proof:
\[
G_T(n) \geq \frac{n}{4}|T_{>n}|\ln{d} - T P\ln{d} + \sum_{|\rho_t|\leq n}{\ln\left( 1 + \frac{\gamma u_t}{1-d^{-\frac{n}{2}}}\right)}.
\]
\qed\end{proof}

We will now show $n=4P$ is a good choice for a concentration controller.

\begin{theorem}\label{thm:optimal_static}
For any sequence of $\rho_t,u_t$ s.t. $P\ln{d}\leq 1$ and $\frac{a}{\gamma}P^2\ln^2{d} \leq u_t \leq \frac{2}{\gamma}$, the $4P$-static strategy reward holds:
\[ G_T(4P) \geq |T_{\leq 4P}| \ln\left(1+\frac{a-2}{4}P\ln{d}\right) \]
\end{theorem}

\begin{proof}

From our assumption on the trading volume and Lemma \ref{lem:static_bound} we get:
\begin{flalign*}
G_T(4P) &\geq |T_{>4P}|P\ln{d} - T P\ln{d} + \sum_{|\rho_t|\leq 4P}{\ln\left( 1 + \frac{\gamma u_t}{1-d^{-2P}} \right)} \\
&= \sum_{|\rho_t|\leq 4P}{\left(\ln\left( 1 + \frac{\gamma u_t}{1-d^{-2P}} \right) - P\ln{d}\right)} \\
&= \sum_{|\rho_t|\leq 4P}{\ln\left(\frac{1-d^{-2P} + \gamma u_t}{d^{P}-d^{-P}}\right)} \\
&\geq |T_{\leq 4P}| \ln\left(\frac{1-d^{-2P} + aP^2\ln^2{d}}{d^{P}-d^{-P}}\right).
\end{flalign*}
Note that $e^x - 1 - e^{-x} + e^{-2x} \leq 2x^2$ for $x \leq 1$. Since we assume $P\ln{d} \leq 1$, we also get that $d^P - 1 - d^{-P} + d^{-2P} \leq 2P^2\ln^2{d}$, and thus,
\begin{flalign*}
G_T(4P) &\geq |T_{\leq 4P}| \ln\left(1+(a-2)\frac{P^2\ln^2{d}}{1-d^{-2P} + 2P^2\ln^2{d}}\right) \\
&\geq |T_{\leq 4P}| \ln\left(1+(a-2)\frac{P^2\ln^2{d}}{1-d^{-2P} + 2P\ln{d}}\right).
\end{flalign*}
Using $d^{-2P} \geq 1-2P\ln{d}$:
\[ G_T(4P) \geq |T_{\leq 4P}| \ln\left(1+\frac{a-2}{4}P\ln{d}\right), \]
concluding our proof.
\qed\end{proof}

\begin{corollary}\label{cor:positive_static}
For any sequence of $\rho_t,u_t$ s.t. $\frac{2}{\gamma}P^2\ln^2{d} \leq u_t \leq \frac{2}{\gamma}$, the $4P$-static strategy reward is positive:
\[ G_T(4P) \geq 0. \]
\end{corollary}

Compare this result to the $\infty$-static strategy. On the one hand, we assume a lower bound on the relative trading volume and not on the average trading volume. However, the required magnitude is smaller by a factor of $2P\ln{d}$, which is usually significantly smaller than $1$.


\section{Adaptive Strategy}\label{sec:adaptive_strategy}
We saw that static strategies can induce positive rewards. However, we don't know the average logarithmic price change $P$ in advance, so we cannot choose the appropriate concentration controller. We solve this problem by using a regret-minimization adaptive algorithm called Exponential Weights Algorithm (EWA), described in Algorithm \ref{alg:EWA}, to gain a reward similar to the optimal static strategy. Essentially, we are learning the optimal concentration controller adaptively, trying to maximize our reward while doing so.

\begin{algorithm}
\caption{\label{alg:EWA} Exponential Weights Adaptive Strategy}
\begin{algorithmic}[1]
\Parameters{$N\in \mathrm{N}, \eta>0$.}
\Init{$r_0(n)\leftarrow 0$ for all $1\leq n\leq N$.}

\For{$1 \le t\le T$}

\State{Set:
\[ p_{t}\left(n\right)\leftarrow\frac{\exp\left(\eta\sum_{0\leq\tau\leq t}{r_\tau(n)}\right)}{\mathlarger{\sum}_{1\leq \mu\leq N}{\exp\left(\eta\sum_{0\leq\tau\leq t}{r_\tau(\mu)}\right)}} \]
for all $1\leq n\leq N$.}

\State{For each concentration controller $1\leq n\leq N$, provide liquidity using $p_{t}\left(n\right)$ of the funds, and observe the current reward $r_t(n)$.}

\EndFor

\end{algorithmic}
\end{algorithm}

Since we saw $4P$ would be a good concentration controller, we want to ensure it is included in our options. Hence, we need to choose $N$ such that $N\geq 4P$. It is reasonable to assume $d^P \leq 2$, so it is enough to choose $N\geq \log_d{16}$. However we need $N$ to be a natural number, and because $d<2$, we can choose $N=\lfloor \log_d{32}\rfloor$.

Now, before we start the analysis, we need to bound the reward:
\begin{lemma}\label{lem:bound_reward}
Choose $N=\lfloor \log_d{32}\rfloor$. For any sequence of $\rho_t,u_t$ s.t. $d^{|\rho_t|} \leq 2$ and $u_t \leq \frac{2}{\gamma}$, we have:
\[R\triangleq\max_{t\leq T, n \leq N}{r_t(n)} - \min_{t\leq T, n \leq N}{r_t(n)} \leq 16\]
\end{lemma}
\begin{proof}
We evaluate each case in Eq. \ref{eq:reward}.

In the case that $\rho_t > n$, we have $r_t(n)=\ln{\frac{1 + d^{\frac{n}{2}}}{2}}$. It is easy to see that $r_t(n) \geq 0$, and also:
\[ r_t(n) \leq \ln{\frac{1 + d^{\frac{N}{2}}}{2}} \leq \ln{\frac{1 + d^{\frac{\log_d{32}}{2}}}{2}} = \ln{\frac{1+\sqrt{32}}{2}} \leq 2. \]

In the case that $\rho_t < -n$, we have $r_t(n)=\ln{\frac{d^{\rho_t} (1 + d^{\frac{n}{2}})}{2}}$. Since $\rho_t<-n$ we get $r_t(n)\leq 0$. Using our assumption that $|\rho_t| \leq \log_d{2}$, we get:
\[ r_t(n) \geq \rho_t\ln{d} \geq -\ln{2}. \]

Otherwise, if $|\rho_t| < n$, we have $r_t(n) = \ln{\frac{2d^\frac{\rho_t}{2} - d^{-\frac{n}{2}} (1 + d^{\rho_t}) + 2\gamma u_t}{2 ( 1- d^{-\frac{n}{2}} )}}$. For a lower bound we can set $u_t=0$, and get from Eq. \ref{eq:term_bound} that:
\[ r_t(n) \geq \ln{d^{-|\rho_t|}} = -|\rho_t|\ln{d} \geq -\ln{2}. \]
For an upper bound, we use the assumption that $u_t\leq \frac{2}{\gamma}$ and again use Eq. \ref{eq:term_bound}:
\[
r_t(n) \leq \ln{\frac{d^{-|\rho_t|}(d^{\frac{3n}{2}} - d^{\frac{n}{2}}) + 4}{2 ( 1- d^{-\frac{n}{2}} )}} \\
\leq \ln{\frac{d^{\frac{3n}{2}} + 4}{2 ( 1- d^{-\frac{n}{2}} )}},
\]
and since $1 \leq n \leq N \leq \log_d{32}$ and $d\geq 1.0001$:
\[ r_t(n) \leq \ln{\frac{d^{\frac{3N}{2}} + 4}{2 ( 1- d^{-\frac{1}{2}} )}} \leq \ln{\frac{32^{\frac{3}{2}} + 4}{2 ( 1 - 1.0001^{-\frac{1}{2}} )}} \leq 15. \]

Overall, we get that $-\ln{2} \leq r_t(n) \leq 15$ for all $t,n$, and thus
\[ R \leq 15 + \ln{2} \leq 16\]
as desired.
\qed\end{proof}

We will start the analysis by presenting a well-known Lemma \cite{cesa1997use,cesa2006prediction}, bounding the \textit{regret} of our adaptive strategy:
\begin{lemma}\label{lem:EWA_regret}
 The rewards of the Exponential Weights Adaptive Strategy hold:
\[ \sum_{t\leq T}{\sum_{1\leq \mu \leq N}{p_t(\mu)r_t(\mu)}} \geq G_T(n)-\frac{\ln{N}}{\eta}-\frac{\eta}{2} T R^2 \]
For all $1\leq n\leq N$, where $R$ is a bound on the range of the reward.
\end{lemma}

Using this lemma, we can now get the main result of this work:

\begin{theorem}\label{thm:adaptive_bound}
Choose $N=\lfloor \log_d{32}\rfloor$ and $\eta = \sqrt{\frac{\ln{N}}{128T}}$. For any sequence of $\rho_t,u_t$ s.t. $d^{|\rho_t|} \leq 2$ and $\frac{10}{\gamma}P^2\ln^2{d} \leq u_t \leq \frac{2}{\gamma}$, the reward of the Exponential Weights Adaptive Strategy (Algorithm \ref{alg:EWA}), $G_T^{(\mathrm{EWA})}$, holds:
\[ G_T^{(\mathrm{EWA})} \geq \frac{3}{4}TP\ln{d} -23\sqrt{T\ln{\log_d{32}}}. \]
\end{theorem}

\begin{proof}
Denote by $M_{t+1}(n)$ the value of the portfolio at $t+1$ if we were to provide liquidity at step $t$ using only the concentration controller $n$, such that $r_t(n) = \ln{\frac{M_{t+1}(n)}{M_t}}$. Thus, due to Jensen's inequality:
\begin{flalign*}
G_T^{(\mathrm{EWA})} &= \sum_{t\leq T}{\ln{\left(\sum_{1\leq n \leq N}{\frac{p_t(n)M_{t+1}(n)}{M_{t}}}\right)}} \\
&\geq \sum_{t\leq T}{\sum_{1\leq n \leq N}{p_t(n)\ln{\frac{M_{t+1}(n)}{M_{t}}}}} \\
&= \sum_{t\leq T}{\sum_{1\leq n \leq N}{p_t(n)r_t(n)}}.
\end{flalign*}

We can now use this result with Lemma \ref{lem:EWA_regret}. We substitute the $\eta$ we chose, set $n=4P \leq N$ and use $R$ from Lemma \ref{lem:bound_reward} to get:
\[
G_T^{(\mathrm{EWA})} \geq G_T(4P)-\sqrt{512T\ln{N}}
\geq G_T(4P)-23\sqrt{T\ln{N}}.
\]
From Theorem \ref{thm:optimal_static}, using $a=10$, we have:
\[
G_T^{(\mathrm{EWA})} \geq |T_{\leq 4P}|\ln{\left(1 + 2P\ln{d} \right)} - 23\sqrt{T\ln{N}} \geq |T_{\leq 4P}|P\ln{d} - 23\sqrt{T\ln{N}}.
\]
Using Markov's Inequality, we get $|T_{\leq 4P}| \geq \frac{3}{4}T$. Thus,
\[
G_T^{(\mathrm{EWA})} \geq \frac{3}{4}TP\ln{d} - 23\sqrt{T\ln{N}} \geq \frac{3}{4}TP\ln{d} -23\sqrt{T\ln{\log_d{32}}}
\]
as desired.
\qed\end{proof}

We can now see that for a large enough time period, the reward will be positive:

\begin{corollary}\label{cor:adaptive_positive}
Choose $N=\lfloor \log_d{32}\rfloor$ and $\eta = \sqrt{\frac{\ln{N}}{128T}}$. For any sequence of $\rho_t,u_t$ s.t. $d^{|\rho_t|} \leq 2$ and $\frac{10}{\gamma}P^2\ln^2{d} \leq u_t \leq \frac{2}{\gamma}$, the reward of the Exponential Weights Adaptive Strategy is positive for all sufficiently large $T$:
\[ G_T^{(\mathrm{EWA})} \geq 0. \]
\end{corollary}

\section{Simulations}\label{sec:simulations}
We collected trading data from two of the largest Uniswap v3 liquidity pools on Ethereum. The first is USDT/USDC\footnote{Contract address: 0x3416cF6C708Da44DB2624D63ea0AAef7113527C6} with 0.01\% trading fee. This is an example of a \textit{stable coin} pool of two tokens representing a single US dollar and should not change in value. This pool is an example of low price volatility and low trading fee. The second liquidity pool is ETH/USDC\footnote{Contract address: 0x8ad599c3A0ff1De082011EFDDc58f1908eb6e6D8} with 0.3\% trading fee, which is an example of a pool with large price volatility, compensated by a high trading fee. In both pools, we look to maximize our value in USDC.

We simulated online static strategies across a wide range of concentration controllers, with hourly re-investments. We also used results from the static strategies to simulate the Exponential Weights Adaptive Strategy. All data were collected using a public Ethereum dataset on Google BigQuery across 2022, and our source code is publicly available at \href{https://github.com/yogi-bo/uniswap-v3-online-framework}{https://github.com/yogi-bo/uniswap-v3-online-framework}.

\subsection{Prices and Volumes}
We are interested to see how our assumptions hold on real-world data. Table \ref{table:real_data} summarizes the empirical data we gathered using the usual notation, and figures \ref{fig:usdt_usdc_data},\ref{fig:eth_usdc_data} shows the price changes and trading volumes across the entire year.

\begin{table}[h]
\centering
\begin{tabular}{ |c||c|c|c|c|c|c| }
 \hline
 Liquidity pool & $d$ & $\gamma$ & $\max_t{d^{|\rho_t|}}$ & $P$ & $\frac{P\ln{d}}{\gamma}$ & $\frac{1}{T}\sum_{t=1}^T{u_t}$ \\
 \hline
 USDT/USDC & $1.0001$ & $10^{-4}$ & $1.0099$ & $2.4\cdot 10^{-1}$ & $2.4\cdot 10^{-1}$ & $3.48\cdot 10^{-1}$ \\
 \hline
 ETH/USDC &  $1.00602$ & $3\cdot 10^{-3}$ & $1.11$ & $9.22\cdot 10^{-1}$ & $1.84$ & $3.62\cdot 10^{-3}$ \\
 \hline
\end{tabular}
\caption{Summary of real-world data for hourly time steps during 2022.}
\label{table:real_data}
\end{table}

First, we can tell that the assumption that $d^{|\rho_t|}<2$ is supported by the empirical data. Also, recall that $\frac{P\ln{d}}{\gamma}$ acts as a benchmark for the required trading volume for a positive reward. In USDT/USDC, we can see that indeed the average relative trading volume is larger, so we expect a positive reward. However, in ETH/USDC this is not true, so we don't expect to do well on this pool.

\begin{figure}
\centering
\includegraphics[width=.8\textwidth]{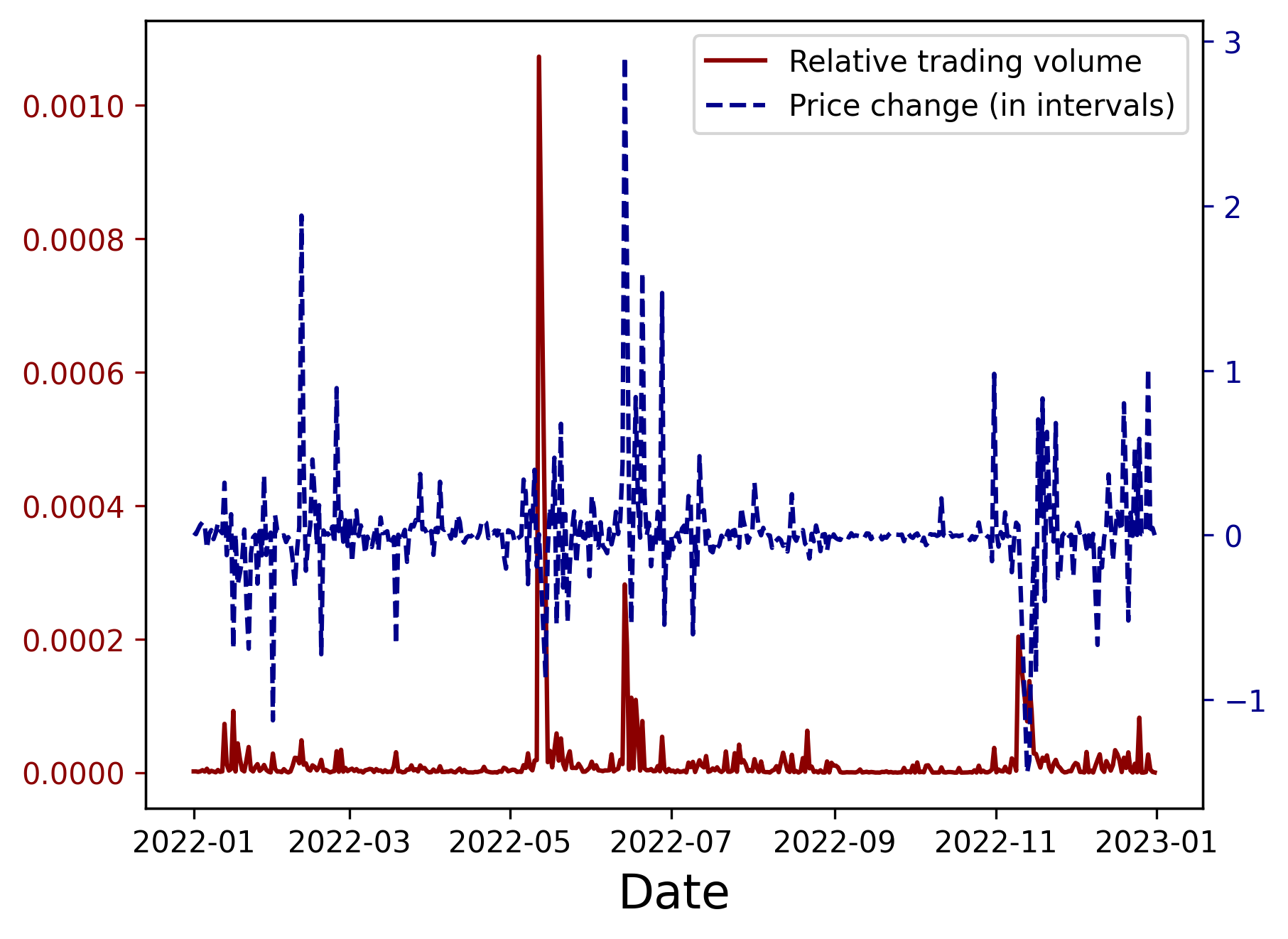}
\caption{USDT/USDC hourly price changes and relative trading volumes during 2022.}
\label{fig:usdt_usdc_data}
\end{figure}

\begin{figure}
\centering
\includegraphics[width=.8\textwidth]{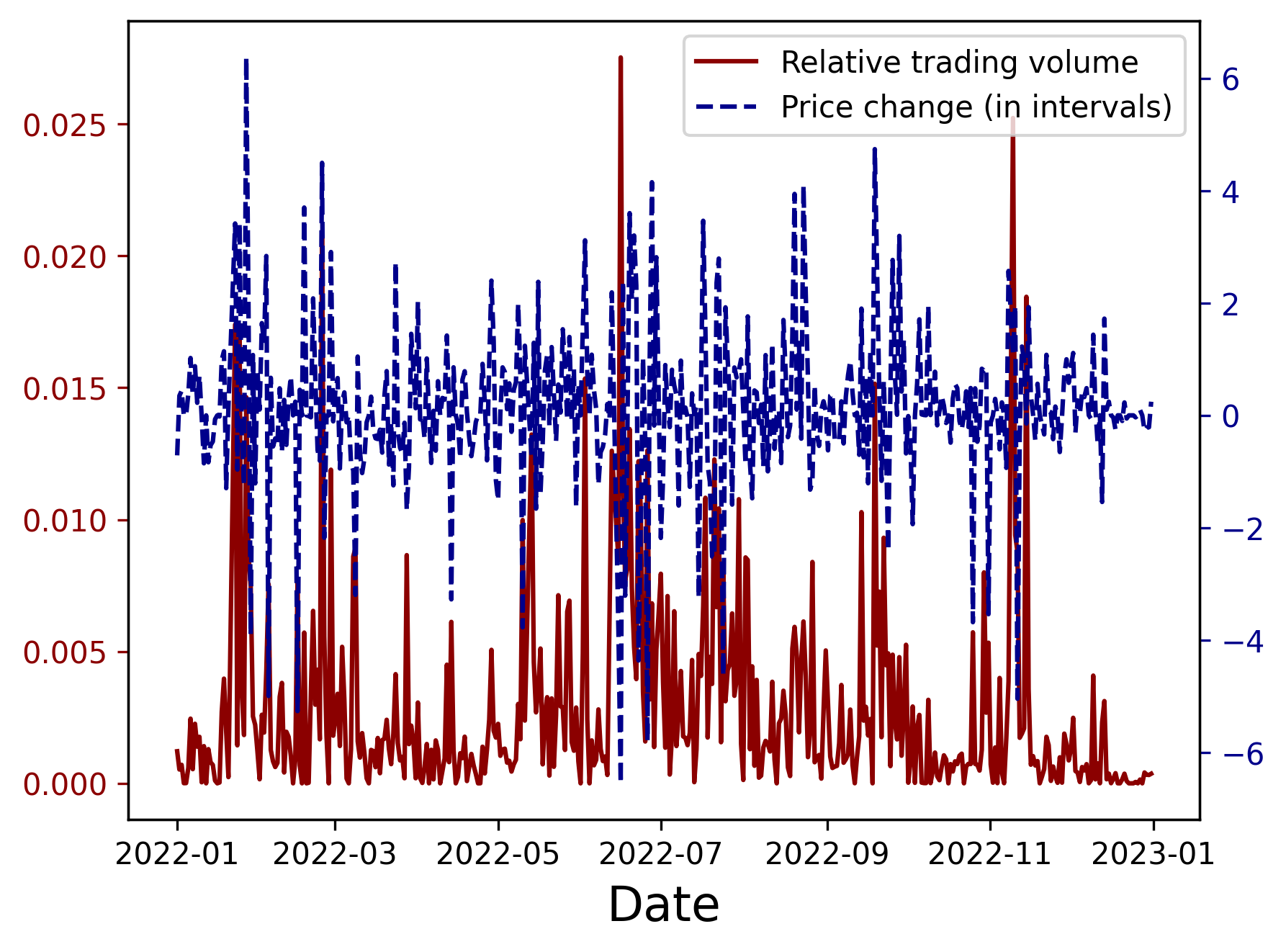}
\caption{ETH/USDC hourly price changes and relative trading volumes during 2022.}
\label{fig:eth_usdc_data}
\end{figure}

\subsection{Rewards}
We simulated the static strategy with exponentially growing concentration controllers. We also ran the Exponential Weights Adaptive Strategy using the simulated rewards of the static strategies (with an adjusted learning rate of $\eta = 10^3$). Table \ref{table:rewards} shows the total reward of each concentration controller, together with the reward of EWA. Figures \ref{fig:usdt_usdc_reward},\ref{fig:eth_usdc_reward} show the rewards across the entire year for chosen strategies.

\begin{table}[h]
\makebox[\textwidth][c]{
\begin{tabular}{ |c||c|c|c|c|c|c|c|c| }
\hline
\multicolumn{9}{|c|}{USDT/USDC} \\
\hline
Strategy & $n=10^0$ & $n=10^1$ & $n=10^2$ & $n=10^3$ & $n=10^4$ & $n=10^5$ & $n=10^6$ & EWA \\
\hline
Reward & $1.04\cdot 10^{-1}$ & $9.048\cdot 10^{-3}$ & $2.14\cdot 10^{-2}$ & $5.63\cdot 10^{-1}$ & $1.47\cdot 10^{-1}$ & $1.77\cdot 10^{-1}$ & $1.76\cdot 10^{-1}$ & $5.15\cdot 10^{-1}$\\
\hline
\end{tabular}
}

\bigskip

\makebox[\textwidth][c]{
\begin{tabular}{ |c||c|c|c|c|c|c|c|c| }
\hline
\multicolumn{9}{|c|}{ETH/USDC} \\
\hline
Strategy & $n=4^0$ & $n=4^1$ & $n=4^2$ & $n=4^3$ & $n=4^4$ & $n=4^5$ & $n=4^6$ & EWA \\
\hline
Reward & $-9.42\cdot 10^{-1}$ & $-6.01\cdot 10^{-1}$ & $-4.38\cdot 10^{-1}$ & $-4.5\cdot 10^{-1}$ & $-4.70\cdot 10^{-1}$ & $-4.73\cdot 10^{-1}$ & $-4.73\cdot 10^{-1}$ & $-4.86\cdot 10^{-1}$\\
\hline
\end{tabular}
}

\caption{Total reward for various static strategies and EWA during 2022.}
\label{table:rewards}
\end{table}

As we expected, the reward for USDT/USDC is positive, while the reward for ETH/USDC is negative, and EWA admits a reward close to the best static controller. We note that artificially increasing the trading volume by $10\%$ was enough for EWA to admit a positive reward for ETH/USDC as well.

\begin{figure}
\centering
\includegraphics[width=.8\textwidth]{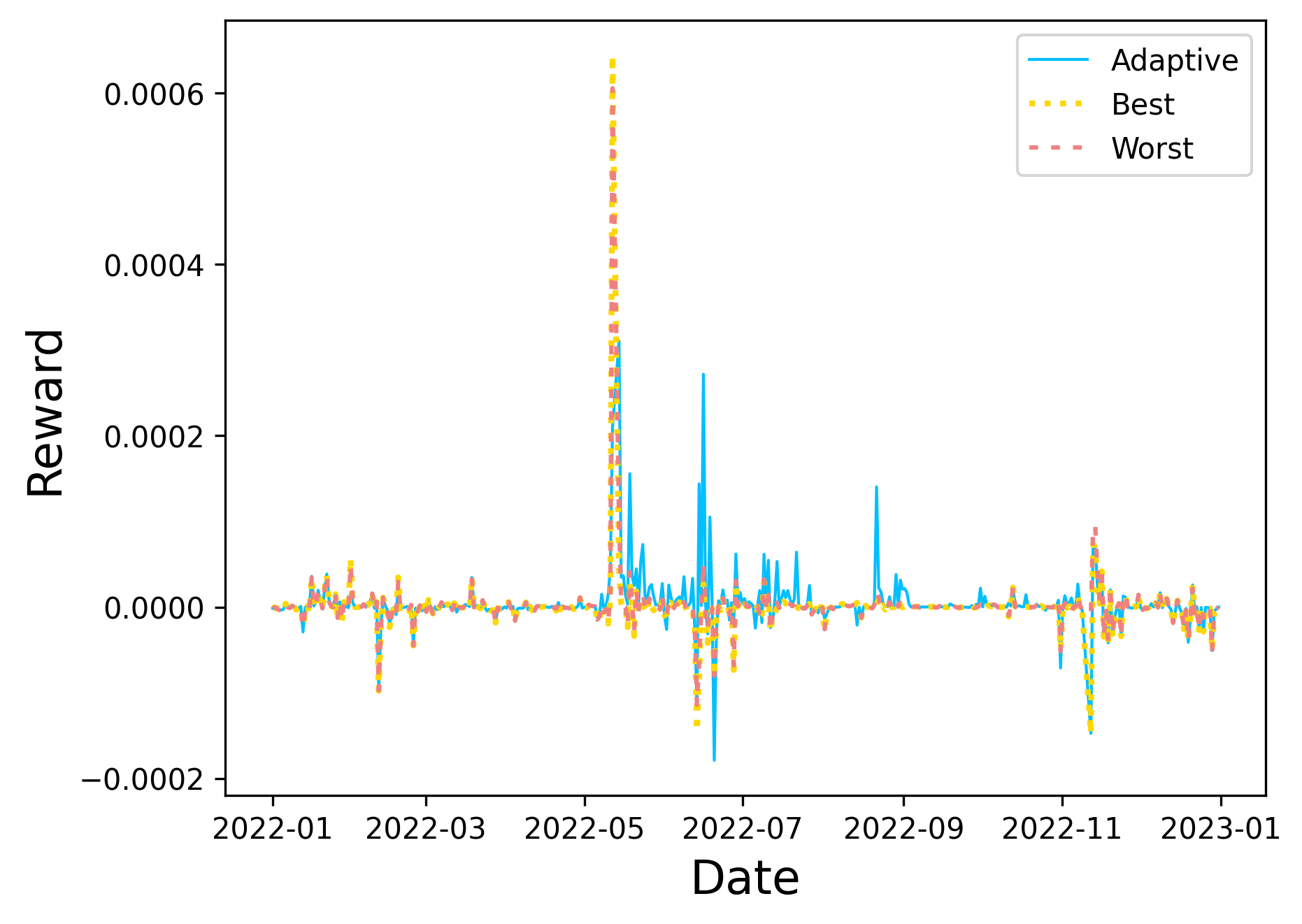}
\caption{USDT/USDC reward for best, worst, and adaptive concentration controllers.}
\label{fig:usdt_usdc_reward}
\end{figure}

\begin{figure}
\centering
\includegraphics[width=.8\textwidth]{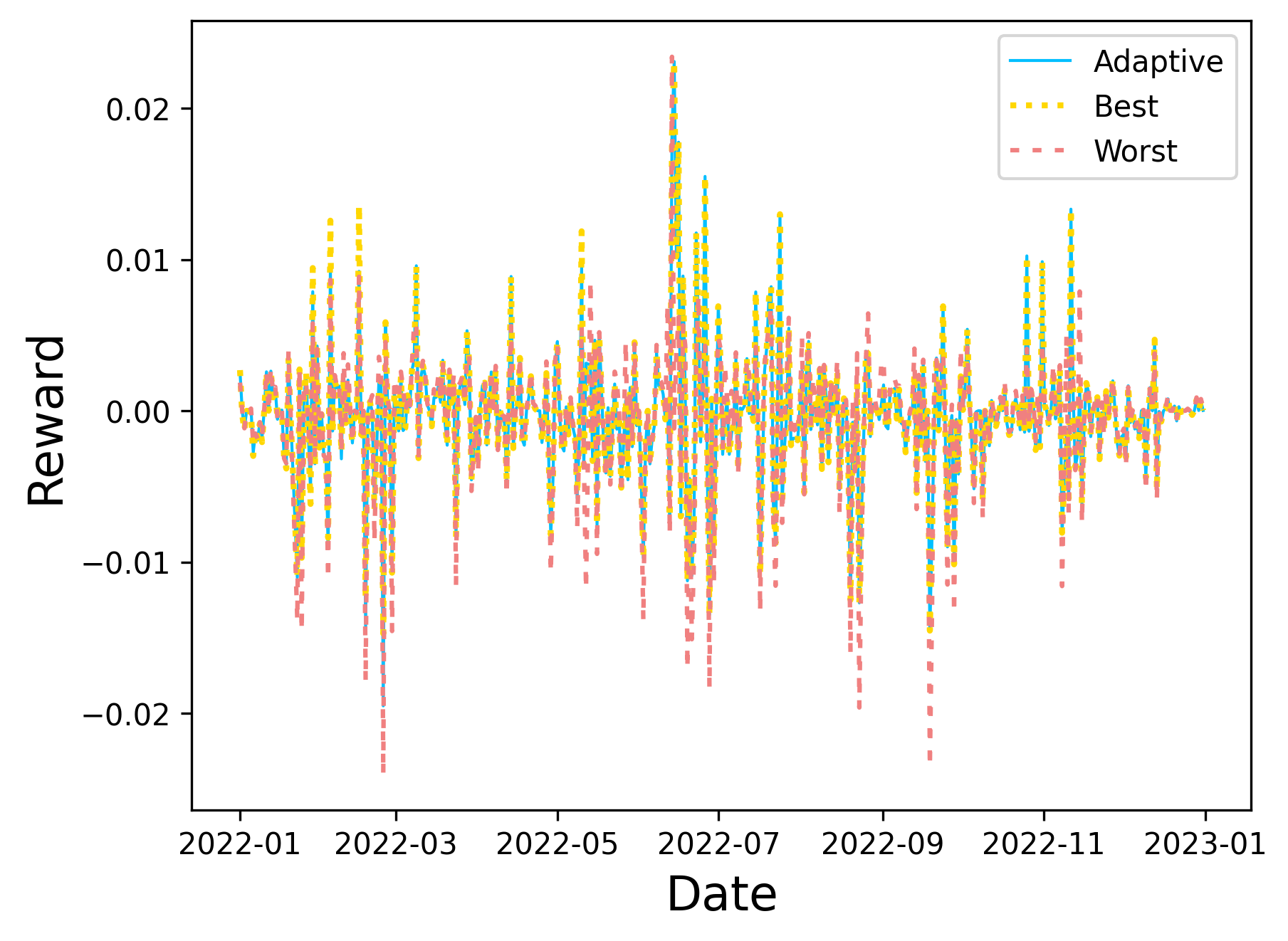}
\caption{ETH/USDC reward for best, worst, and adaptive concentration controllers.}
\label{fig:eth_usdc_reward}
\end{figure}


\section{Conclusions}\label{sec:conclusions}
Uniswap v3 introduced a novel Automated Market Maker mechanism that allows for complicated investment strategies for liquidity providers. Our work demonstrates the effectiveness of using online learning methods to analyze and create such strategies. We modeled the problem of Uniswap v3 liquidity provision as an online learning problem and analyzed the reward in the context of the new formalism. In our main result (Theorem \ref{thm:adaptive_bound}), we showed a dynamic strategy based on prediction with expert advice. We proved this strategy admits a positive total reward for a large enough trading volume, even in the face of non-stochastic price changes. We present the needed trading volume for a positive reward in terms of the average logarithmic price change.

Furthermore, using empirical data from real Uniswap liquidity pools, we showed the online learning model is a good approximation for the liquidity provision problem and can be used in real-life scenarios.

Future work could improve our analysis by improving the accuracy of our formal reward function. Such improvements may include partial trading fees for steps with a significant price change or gas fees for reinvestment.

\section*{Acknowledgments}
This project has received funding from the European Research Council (ERC) under the European Union’s Horizon 2020 research and innovation program (grant agreement No. 882396), the Israel Science Foundation (grant number 993/17), the Yandex Initiative for Machine Learning at Tel Aviv University and a grant from the Tel Aviv University Center for AI and Data Science (TAD).

\bibliography{main}

\end{document}